\documentclass{elsart}

\usepackage{amsmath}
\usepackage{amsthm}
\usepackage{graphicx}

\newtheorem{theorem}{Theorem}
\newtheorem{lemma}{Lemma}

\renewcommand{\log}{\lg}
\newcommand{\lookup}
    {\ensuremath{\mathsf{lookup}}}
\newcommand{\nocc}
    {\ensuremath{\mathit{nocc}}}
\newcommand{\ndoc}
    {\ensuremath{\mathit{ndoc}}}

\newcommand{\rank}
    {\ensuremath{\mathsf{rank}}}

\newcommand{\search}
    {\ensuremath{\mathsf{search}}}
\newcommand{\select}
    {\ensuremath{\mathsf{select}}}

\newtheorem{example}{Example}

\newcommand{\ejemplo}[1]{}

\begin{document}

\begin{frontmatter}

\title{Document Listing on Repetitive Collections with Guaranteed Performance
\thanksref{CPM}}

\author{Gonzalo Navarro}
\address{Center for Biotechnology and Bioengineering (CeBiB),
Dept. of Computer Science, University of Chile, 
Beauchef 851, Santiago, Chile.
  \texttt{gnavarro@dcc.uchile.cl}}

\thanks[CPM]{Supported in part by Fondecyt grant 1-170048 and Basal Funds 
FB0001, Conicyt, Chile. A preliminary version of this article appeared in 
{\em Proc.  CPM 2017}.}


\begin{abstract}
We consider document listing on string collections, that is, finding in which 
strings a given pattern appears. In particular, we focus on repetitive 
collections: a collection of size $N$ over alphabet $[1,\sigma]$ is composed 
of $D$ copies of a string of size $n$, and $s$ 
edits are applied on ranges of copies. We introduce the first document listing 
index with size $\tilde{O}(n+s)$, precisely $O((n\log\sigma+s\log^2 N)\log D)$ 
bits, and with useful worst-case time guarantees: Given a pattern of length $m$,
the index reports the $\ndoc>0$ strings where it appears in time 
$O(m\log^{1+\epsilon} N \cdot \ndoc)$, for any constant $\epsilon>0$ (and 
tells in time $O(m\log N)$ if $\ndoc=0$). Our technique is to augment a range 
data structure that is commonly used on grammar-based indexes, so that instead 
of retrieving all the pattern occurrences, it computes useful summaries on 
them. We show that the idea has
independent interest: we introduce the first grammar-based index that, on
a text $T[1,N]$ with a grammar of size $r$, uses $O(r\log N)$ bits and 
counts the number of occurrences of a pattern $P[1,m]$ in time 
$O(m^2 + m\log^{2+\epsilon} r)$, for any constant $\epsilon>0$. We also give
the first index using $O(z\log(N/z)\log N)$ bits, where $T$ is parsed by
Lempel-Ziv into $z$ phrases, counting occurrences in time $O(m\log^{2+\epsilon} N)$.


\end{abstract}

\end{frontmatter}

\section{Introduction} \label{sec:introduction}

Document retrieval on general string collections is an area that has recently
attracted attention \cite{Nav14}. On the one hand, it is a natural
generalization of the basic Information Retrieval tasks 
carried out on search engines \cite{BYRN99,BCC10}, many of which are also useful
on Far East languages, collections of genomes, code repositories, multimedia 
streams, etc. It also enables phrase queries on natural language
texts. On the other hand, it raises a number of algorithmic 
challenges that are not easily addressed with classical pattern matching 
approaches.

In this paper we focus on one of the simplest document retrieval problems, 
{\em document listing} \cite{Mut02}. Let $\mathcal{D}$ be a collection of
$D$ documents of total length $N$. We want to build an index on $\mathcal{D}$
such that, later, given a search pattern $P$ of length $m$, we report the
identifiers of all the $\ndoc$ documents where $P$ appears. Given that $P$ may
occur $\nocc \gg \ndoc$ times in $\mathcal{D}$, resorting to pattern matching,
that is, finding all the $\nocc$ occurrences and then listing the distinct
documents where they appear, can be utterly inefficient. Optimal $O(m+\ndoc)$
time  document listing solutions appeared only in 2002 \cite{Mut02}, 
although they use too much space. There are also more recent statistically 
compressed indices \cite{Sad07,HSV09}, which are essentially space-optimal
with respect to the statistical entropy and pose only a small time penalty.

We are, however, interested in {\em highly repetitive} string
collections \cite{Naviwoca12}, which are formed by a few distinct 
documents and a number of near-copies of those. Such collections arise, for 
example, when sequencing the genomes of thousands of individuals of a few 
species, when managing versioned collections of documents like Wikipedia, and 
in versioned software repositories. 
Although many of the fastest-growing datasets are indeed repetitive, this is
an underdeveloped area: most succinct indices for string collections are based 
on statistical compression, and these fail to exploit repetitiveness \cite{KN13}.

\subsection{Modeling repetitiveness}

There are few document listing indices that profit from repetitiveness.
A simple model to analyze them is as follows \cite{MNSV10,GKNPS13,Naviwoca12}: 
Assume there is a single document of size $n$ on alphabet $[1,\sigma]$, and 
$D-1$ copies of it, on which $s$ single-character edits (insertions, deletions,
substitutions) are distributed arbitrarily,
forming a collection of size $N \approx nD$. This models, for 
example, collections of genomes and their single-point mutations. 
For versioned documents and software repositories, a better model is a 
generalization where each edit affects a range of copies, such as an
interval of versions if the collection has a linear versioning structure, or
a subtree of versions if the versioning structure is hierarchical.

The gold standard to measure space usage on repetitive collections is the size 
of the {\em Lempel-Ziv parsing} \cite{LZ76}. If we parse the concatenation of 
the strings in a repetitive collection under either of the models above, we 
obtain at most $z = O(n/\log_\sigma n + s) \ll N$ phrases. 
Therefore, while a statistical compressor would require basically $N\log\sigma$
bits if the base document is incompressible \cite{KN13}, we can aim to reach as 
little as $O(n\log\sigma + s\log N)$ bits by expoiting repetitiveness via
Lempel-Ziv compression (an arbitrary Lempel-Ziv pointer 
requires $O(\log N)$ bits, but those in the first document could use
$O(\log n)$).

This might be too optimistic for an index, however, as there is no known way 
to extract substrings efficiently from Lempel-Ziv compressed text. Instead, 
{\em grammar compression} allows extracting any text symbol in logarithmic time
using $O(r\log N)$ bits, where $r$ is the size of the grammar 
\cite{BLRSRW15,VY13}. 
It is possible to obtain a grammar of size $r=O(z\log(N/z))$ 
\cite{CLLPPSS05,HLR16}, which using standard methods \cite{Ryt03} can be 
tweaked to $r=n/\log_\sigma N + s\log N$ under our repetitiveness model. 
Thus the space we might aim at for indexing is 
$O(n\log\sigma + s\log^2 N)$ bits.

\subsection{Our contributions}

Although they perform reasonably well in practice, none of the existing 
structures for document listing on repetitive collections \cite{CM13,GKNPS13} 
offer good worst-case time guarantees combined with worst-case space guarantees
that are appropriate for repetitive collections, that is, growing with $n+s$ 
rather than with $N$. 
In this paper we present the {\em first document listing index offering good 
guarantees in space and time for repetitive collections}: our index

\begin{enumerate}
	\item uses $O((n\log\sigma + s\log^2 N)\log D)$ bits of space,
and
	\item performs document listing in time $O(m\log^{1+\epsilon} N 
\cdot\ndoc)$, for any constant $\epsilon>0$.
\end{enumerate}

That is, at the price of being an $O(\log D)$ space factor away from what 
could be hoped from a grammar-based index, our index offers document listing
with useful time bounds per listed document.
The result is summarized in Theorem~\ref{thm:main}.

We actually build on a grammar-based document listing index \cite{CM13} that
stores the lists of the documents where each nonterminal appears, and 
augment it by rearranging the nonterminals in different orders, following 
a wavelet tree \cite{GGV03} deployment that guarantees that only $O(m\log r)$ 
ranges of lists have to be merged at query time. 
We do not store the lists themselves in
various orders, but just succinct range minimum query (RMQ) data structures 
\cite{FH11} that allow implementing document listing on ranges of lists
\cite{Sad07}. Even those RMQ structures are too large for our purposes, so 
they are further compressed exploiting the fact that their 
underlying data has long increasing runs, so the structures are reduced with
techniques analogous to those developed for the ILCP data structure 
\cite{GKNPS13}. 

The space reduction brings new issues, however, because we 
cannot afford storing the underlying RMQ sequences. These problems are 
circumvented with a new, tailored, technique to extract the distinct elements
in a range that might have independent interest (see Lemma~\ref{lem:alg} in
Appendix~\ref{sec:app}).

\paragraph*{Extensions}
The wavelet tree \cite{GGV03} represents a two-dimensional grid with points.
It is used in grammar-based indexes \cite{CNfi10,CN12,CE18} to enumerate all the
occurrences of the pattern: a number of {\em secondary} occurrences are 
obtained from each point that qualifies for the query. At a high level, our
idea above is to compute {\em summaries} of the qualifying points instead of 
enumerating them one by one. We show that this idea has independent interest
by storing the number of secondary occurrences that can be obtained from each 
point. The result is an index of $O(r\log N)$ bits, similar to the size of
previous grammar-based indexes \cite{CNfi10,CN12,CE18}, and able to 
count the number of occurrences of the pattern in time 
$O(m^2 + m\log^{2+\epsilon} r)$ for any constant $\epsilon>0$ (and
$O(m(\log N + \log^{2+\epsilon} r))$ if the grammar is balanced); see
Theorem~\ref{thm:count}.
Current grammar-based indexes are unable to count
the occurrences without locating them one by one, so for the first time
a grammar-based index can offer efficient counting.
Further, by using recent techniques \cite{GNP18}, we also obtain the improved
time $O(m\log N + \log^\epsilon r \cdot \nocc)$ for an index based on a
balanced grammar that reports the $\nocc$ occurrences; see 
Theorem~\ref{thm:locate}. Indeed, Lempel-Ziv based indexes are also unable
to count without locating. As a byproduct of our counting grammar-based index,
we obtain a structure of $O(z\log(N/z)\log N)$ bits, where $z \le r$ is the 
size of the Lempel-Ziv parse of $T$, that can count in time 
$O(m\lg^{2+\epsilon} N)$; see Theorem~\ref{thm:countz}.

As another byproduct, we improve an existing result \cite{NNR13} on computing 
summaries of two-dimensional points in ranges, when the points have associated
values from a finite group. We show in Theorem~\ref{thm:sum} that, within 
linear space, the time to
operate all the values of the points in a given range of an
$r \times r$ grid can be reduced from $O(\log^3 r)$ to 
$O(\log^{2+\epsilon} r)$, for any constant $\epsilon>0$.

\section{Related work}

The first optimal-time and linear-space solution to document listing is due to
Muthukrishnan \cite{Mut02}, who solves the problem in $O(m+\ndoc)$ time using
an index of $O(N\log N)$ bits of space. Later solutions \cite{Sad07,HSV09}
improved the space to essentially the statistical entropy of $\mathcal{D}$,
at the price of multiplying the times by low-order polylogs of $N$ (e.g.,
$O(m+\log N \cdot \ndoc)$ time with $O(N)$ bits on top of the entropy
\cite{Sad07,BN13}). However, statistical entropy does not capture repetitiveness
well \cite{KN13}, and thus these solutions are not satisfactory in repetitive
collections.

There has been a good deal of work on pattern matching indices for repetitive 
string collections \cite[Sec 13.2]{Nav16}: building on 
regularities of suffix-array-like structures 
	\cite{MNSV10,NPCHIMP13,NPLHLMP13,BCGPR15,GNP18}, 
on grammar compression \cite{CNfi10,CN12,CE18}, 
on variants of Lempel-Ziv compression \cite{KN13,GGKNP12,DJSS14,BEGV18}, 
and on combinations \cite{GGKNP12,GGKNP14,HLSTY10,YWLWX13,BGGMS14,NIIBT16,NP18}.
However, there has been little work on document retrieval structures for 
repetitive string collections.

One precedent is Claude and Munro's index based on grammar compression
\cite{CM13}. It builds on a grammar-based pattern-matching index \cite{CN12}
and adds an {\em inverted index} that explicitly indicates the documents where 
each nonterminal appears; this inverted index is also grammar-compressed. To 
obtain the answer, an unbounded number of those lists of documents must be 
merged. No relevant worst-case time or space guarantees are offered.

Another precedent is ILCP \cite{GKNPS13}, where it is shown that an array 
formed by interleaving the longest common prefix arrays of the documents in
the order of the global suffix array, ILCP, has long increasing runs on
repetitive collections.
Then an index of size bounded by the runs in the suffix array \cite{MNSV10}
and in the ILCP array performs document listing in time $O(\search(m) +
\lookup(N) \cdot \ndoc)$, where $\search$ and $\lookup$ are the search and 
lookup time, respectively, of a run-length compressed suffix array 
\cite{MNSV10,GNP18}. Yet, there are only average-case bounds for the size of the
structure in terms of $s$: If the base document is generated at random and 
the edits are spread at random, then the structure uses
$O(n\log N + s\log^2 N)$ bits on average. 

The last previous work is PDL \cite{GKNPS13}, which stores inverted lists at
sampled nodes in the suffix tree of $\mathcal{D}$, and then grammar-compresses
the set of inverted lists. For a sampling step $b$, it requires 
$O((N/b)\log N)$ bits plus the (unbounded) space of the inverted lists. 
Searches that lead to the sampled nodes have their
answers precomputed, whereas the others cover a suffix array range of size
$O(b)$ and are solved by brute force in time $O(b\cdot\lookup(N))$. 

To be fair, those indexes perform well in many practical situations 
\cite{GHKKNPS17}. However, in this article we are interested in whether
providing worst-case guarantees in time and space.

\section{Basic Concepts} \label{sec:related}

\subsection{Listing the different elements in a range} \label{sec:listing}

Let $L[1,t]$ be an array of integers in $[1,D]$.  Muthukrishnan \cite{Mut02}
gives a structure that, given a range $[i,j]$, lists all the $\ndoc$
distinct elements in $L[i,j]$ in time $O(\ndoc)$. He defines an array $E[1,t]$ 
(called $C$ in there) storing in $E[k]$ the largest position $l<k$ where 
$L[l]=L[k]$, or
$E[k]=0$ if no such position exists. Note that the leftmost positions of the
distinct elements in $L[i,j]$ are exactly those $k$ where $E[k] < i$.
He then stores a data structure supporting range-minimum queries (RMQs) on $E$, 
$\textsc{rmq}_E(i,j) = \mathrm{argmin}_{i \le k \le j} E[k]$ \cite{FH11}.
Given a range $[i,j]$, he computes $k = \textsc{rmq}_E(i,j)$. If
$E[k] < i$, then he reports $L[k]$ and continues recursively on $L[i,k-1]$ and 
$L[k+1,j]$. Whenever it turns out that $E[k] \ge i$ for an interval $[x,y]$, 
there are no leftmost occurrences of $L[i,j]$ within $L[x,y]$, so this interval
can be abandoned.
It is easy to see that the algorithm takes $O(\ndoc)$ time and uses $O(t\log t)$
bits of space; the RMQ structure uses just $2t+o(t)$ bits and answers 
queries in constant time \cite{FH11}.

Furthermore, the RMQ structure does not even access $E$. Sadakane \cite{Sad07}
replaces $E$ by a bitvector $V[1,D]$ to mark which elements have been reported.
He sets $V$ initially to all zeros and replaces the test $E[k] < i$ by
$V[L[k]]=0$, that is, the value $L[k]$ has not yet been reported (these tests
are equivalent only if we recurse left and then right in the interval 
\cite{Nav14}). If so, he reports $L[k]$ and sets $V[L[k]] \leftarrow 1$. 
Overall, he needs only $O(t+D)$ bits of space on top of $L$, and still runs in 
$O(\ndoc)$ time ($V$ can be reset to zeros by rerunning the query or through 
lazy initialization). Hon et al.~\cite{HSV09} further reduce the extra space to
$o(t)$ bits, yet increasing the time, via sampling the array $E$.

\ejemplo{

\begin{figure}[t]
\begin{center}
\includegraphics[width=0.5\textwidth]{muthu.pdf}
\end{center}
\caption{An example of the structure to find the distinct elements in a
range. Array $L$ contains the elements, $E$ contains the pointers to previous
occurrences, $F$ marks the run heads in $E$, and $E'$ stores those run heads.}
\label{fig:muthu}
\end{figure}

\bigskip

\begin{example}
Fig.~\ref{fig:muthu} shows an example on an array $L[1,13]$; for now consider
only arrays $L$ and $E$. To find the distinct elements in $L[5,13]$ we start
with $k = \textsc{rmq}_E(5,13) = 7$. Since $E[7]=2<5$ (or, in Sadakane's
version, since $L[7]=2$ has not been reported), we report value $L[7]=2$ and
recurse on $L[5,6]$ and $L[8,13]$. In the first we compute $k =
\textsc{rmq}_E(5,6) = 5$ and, since $E[5]=4<5$ (or, in Sadakane's version,
since $L[5]=1$ has not been reported), we report value $L[5]=1$ and recurse on
$L[6,6]$. But $E[6]=5$ is not less than $5$ (or, in Sadakane's version,
$L[6]=1$ has already been reported), we do not continue. Now returning to
$L[8,13]$, we compute $k=\textsc{rmq}_E(8,13) = 8$. Since $E[8]=3<5$ (or, in
Sadakane's version, since $L[8]=3$ has not been reported), we report $L[8]=3$
and recurse on $L[9,13]$. We compute $k=\textsc{rmq}_E(9,13) = 12$. Since
$E[12]=6$ is not less than $5$ (or, in Sadakane's version, since $L[12]=1$ has
already been reported), we finish.
\end{example}

}

In this paper we introduce a variant of Sadakane's document listing technique
that might have independent interest; see Section~\ref{sec:dlist} and
Lemma~\ref{lem:alg} in Appendix~\ref{sec:app}.

\subsection{Range minimum queries on arrays with runs} \label{sec:rmq}

Let $E[1,t]$ be an array that can be cut into $\rho$ runs of 
nondecreasing values. Then it is possible to solve RMQs in $O(\log\log t)$
time plus $O(1)$ accesses to $E$ using $O(\rho\log(t/\rho))$ bits. The idea is 
that the possible minima (breaking ties in favor of the leftmost) in
$E[i,j]$ are either $E[i]$ or the positions where runs start in the range.
Then, we can use a sparse bitvector $F[1,t]$ marking with $F[k]=1$ the run 
heads. We also define an array $E'[1,\rho]$, so that if $F[k]=1$ then 
$E'[\rank_1(F,k)] = E[k]$, where
$\rank_v(F,k)$ is the number of occurrences of bit $v$ in $F[1,k]$.
We do not store $E'$, but just an RMQ structure
on it. Hence, the minimum of the run heads in $E[i,j]$ can be found by
computing the range of run heads involved, $i'=\rank_1(F,i-1)+1$ and
$j'=\rank_1(F,j)$, then finding the smallest value among them in $E'$ with
$k' = \textsc{rmq}_{E'}(i',j')$, and mapping it back to $E$ with $k =
\select_1(F,k')$, where
$\select_v(F,k')$ is the position of the $k'$th occurrence of bit $v$ in $B$.
Finally, the RMQ answer is either $E[i]$ or $E[k]$, so we
access $E$ twice to compare them. 

\ejemplo{

\bigskip
\begin{example} \label{ex:runlenRMQ}
The array $E$ of Fig.~\ref{fig:muthu} has $\rho=3$ runs, $E[1,6]$, $E[7,11]$, and
$E[12,13]$. The positions of the run heads are marked in $F$, and the smaller
array $E'$ contains the run heads. The answer to $\textsc{rmq}_E(5,13)$ can be
only $5$ (the first element of the interval, not a run head) or the position of
some of the involved run heads, $7$ and $12$. To find the smallest run head,
we compute the appropriate range in $E'$, $i'=\rank_1(F,5-1)+1=2$ and
$j'=\rank_1(F,13)=3$. The actual RMQ structure is built on the much shorter
$E'$, where we find $k' = \textsc{rmq}_{E'}(2,3)=2$. This position is mapped
back to $k=\select_1(F,2)=7$. Then $\textsc{rmq}_E(5,13)$ is either $5$ (the
leftmost element in the range) or $7$ (the smallest involved run head). We
compare $E[5]$ with $E[7]$ and choose the smaller, $E[7]$.
\end{example}

}

This idea was used by Gagie et al.\ \cite[Sec~3.2]{GKNPS13} for runs of equal 
values, but it works verbatim for runs of nondecreasing values. They show how to
store $F$ in $\rho \lg (t/\rho) + O(\rho)$ bits so that it solves rank in
$O(\log\log t)$ time and select in $O(1)$ time, by augmenting a sparse bitvector
representation \cite{OS07}. This dominates the space and time of the whole
structure.

The idea was used even before by Barbay et al.\ \cite[Thm.~2]{BFN12}, for
runs of nondecreasing values. They 
represented $F$ using $\rho\log(t/\rho)+O(\rho)+o(t)$ bits so that the 
$O(\log\log t)$ time becomes $O(1)$, but we cannot afford the $o(t)$
extra bits in this paper. 

\subsection{Wavelet trees} \label{sec:wt}

A wavelet tree \cite{GGV03} is a sequence representation that supports, in
particular, two-dimensional orthogonal range queries \cite{Cha88,Nav12}. 
Let $(1,y_1), (2,y_2), \ldots, (r,y_r)$ be a sequence of points with 
$y_i \in [1,r]$, and let $S=y_1 y_2 \ldots y_r$ be the $y$ coordinates in order.
The wavelet tree is a perfectly balanced binary tree where each node 
handles a range of $y$ values. The root handles $[1,r]$. If a node handles 
$[a,b]$ then its left child handles $[a,\mu-1]$ and its right child handles 
$[\mu,b]$, with $\mu=\lceil(a+b)/2\rceil$. The leaves handle individual $y$ 
values. If a node handles range $[a,b]$, then it represents the subsequence 
$S_{a,b}$ of $S$ formed by the $y$ coordinates that belong to $[a,b]$. 
Thus at each level the
strings $S_{a,b}$ form a permutation of $S$. What is stored for each such
node is a bitvector $B_{a,b}$ so that $B_{a,b}[i]=0$ iff $S_{a,b}[i] < \mu$,
that is, if that value is handled in the left child of the node. Those 
bitvectors are provided with support for rank and select queries.
The wavelet tree has height $\log r$, and its total space requirement for all
the bitvectors $B_{a,b}$ is $r\log r$ bits. The extra structures for rank and
select add $o(r\log r)$ further bits and support the queries 
in constant time \cite{Cla96,Mun96}. 

We will use wavelet trees where there can be more than one point per column,
say $p\ge r$ points in total.
To handle them, we add a bitvector $R[1,p+1] = 1 0^{c_1-1} 1 0^{c_2-1} 
\ldots 1 0^{c_r-1} 1$, if there are $c_j$ points in column $j$. Then any 
coordinate range $[x_1,x_2]$ is mapped to the wavelet tree columns 
$[\select_1(R,x_1),\select_1(R,x_2+1)-1]$. Conversely, a column $j$
returned by the wavelet tree can be mapped back to the correct coordinate
$x = \rank_1(R,j)$. The wavelet tree then representes a string $S$ of 
length $p$ over the alphabet $[1,r]$ using $p\log r + o(p\log r)$ bits, to
which $R$ adds $p+o(p)$ bits to implement rank and select in constant time.

\ejemplo{

\begin{figure}[t]
\begin{center}
\includegraphics[width=\textwidth]{wtree.pdf}
\end{center}
\caption{An example of a wavelet tree on the grid $[1,13] \times [1,7]$,
where the points have labels. The data in gray is conceptual; only the one in
black is actually represented. We give the names of the sequences up to the
second level only to avoid cluttering.}
\label{fig:wtree}
\end{figure}

\bigskip

\begin{example}
Fig.~\ref{fig:wtree} shows a grid of $p=13$ columns and $r=7$
rows. Only the bitvectors $B_{a,b}$ are represented; the strings $S_{a,b}$ are
conceptual. The example shows that we can also associate {\em labels} with
the points, and these induce (virtual) arrays $A_{a,b}$ associated with the
corresponding points in $S_{a,b}$. We do not usually store the arrays $A_{a,b}$
explicitly, but rather some information on them.
\end{example}

}

With the wavelet tree one can recover any $y_i$ value by tracking
it down from the root to a leaf, but let us describe a more general procedure,
where we assume that the $x$-coordinates are already mapped.

\paragraph*{Range queries}
Let $[x_1,x_2] \times [y_1,y_2]$ be a query range. The number of points that
fall in the range can be counted in $O(\log r)$ time as follows. We start at
the root with the range $S[x_1,x_2] = S_{1,r}[x_1,x_2]$. Then we project the
range both left and right, towards $S_{1,\mu-1}[\rank_0(B_{1,r},x_1-1)+1,
\rank_0(B_{1,r},x_2)]$ and $S_{\mu,r}[\rank_1(B_{1,r},x_1-1)+1,
\rank_1(B_{1,r},x_2)]$, respectively, with $\mu=\lceil(r+1)/2\rceil$. If some
of the ranges is empty, we stop the recursion on that node. If the interval
$[a,b]$ handled by a node is disjoint with $[y_1,y_2]$, we also stop. If the
interval $[a,b]$ is contained in $[y_1,y_2]$, then all the points in the $x$
range qualify, and we simply sum the length of the range to the count.
Otherwise, we keep splitting the ranges recursively. It is well known that
the range $[y_1,y_2]$ is covered by $O(\log r)$ wavelet tree nodes, and that
we traverse $O(\log r)$ nodes to reach them (see Gagie et al.~\cite{GNP11} 
for a review of this and more refined properties). If we also want to report all
the corresponding $y$ values, then instead of counting the points found, we 
track each one individually towards its leaf, in $O(\log r)$ time. At the 
leaves, the $y$ values are sorted. 

\ejemplo{

\bigskip
\begin{example}
To count the number of points inside $[3,6] \times [2,6]$, we start with
$S_{1,7}[3,6] = 4,5,3,5$, and project it to $S_{1,3}[2,2] = 3$ and
$S_{4,7}[2,4] = 4,5,5$. The left range, in $S_{1,3}$, is projected to the
right only, because the left node handles $S_{1,1}$, whose $y$ range has no
intersection with the query range $[2,6]$. The right projection is
$S_{2,3}[2,2]=3$. Since the $y$ range of $S_{2,3}$ is contained in that of the
query, $[2,6]$, we already count the point in $S_{2,3}[2,2]$ as belonging to
the result.  The other range, $S_{4,7}[2,4]$, is projected to the left, 
$S_{4,5}[2,4]=4,5,5$, whereas the projection to the right is empty, 
$S_{6,7}[1,0]$. Since the $y$ range of $S_{4,5}$ is completely contained in 
that of the query, we count the $3$ points in $S_{4,5}[2,4]$ as part of the 
query, and answer that there are $4$ points in the range without need of 
tracking those points down, but we may do if we want to find their coordinates.
In total, we traverse at most the $O(\log y)$ maximal nodes that cover the 
query interval $[2,6]:$ $S_{2,3}$, $S_{4,5}$, $S_{6,6}$, and their ancestors.
\end{example}

}

\paragraph*{Faster reporting}
By using $O(r\log r)$ bits, it is possible to track the positions faster in 
upward direction, and associate the values with their root positions.
Specifically, by using $O((1/\epsilon)r\log r)$ bits, one can reach the root 
position of a    
symbol in time $O((1/\epsilon)\log^\epsilon r)$, for any $\epsilon>0$
\cite{Cha88,Nav12}.
Therefore, the $\nocc$ results can be extracted in time $O(\log r +
\nocc\,\log^\epsilon r)$ for any constant $\epsilon$.

\paragraph*{Summary queries}
Navarro et al.~\cite{NNR13} showed how to perform {\em summary} queries on
wavelet trees, that is, instead of listing all the points that belong to a
query range, compute some summary on them faster than listing the points one
by one. For example, if the points are assigned values in $[1,N]$, then one
can use $O(p\log N)$ bits and compute the sum, average, or variance of the 
values associated with points in a range in time $O(\log^3 r)$, or 
their minimum/maximum in $O(\log^2 r)$ time. The idea is to associate with
the sequences $S_{a,b}$ other sequences $A_{a,b}$ storing the values associated
with the corresponding points in $S_{a,b}$, and carry out range queries on the 
intervals of the sequences $A_{a,b}$ of the $O(\log r)$ ranges into which 
two-dimensional queries are decomposed, in order to compute the desired 
summarizations. To save space, the explicit sequences $A_{a,b}$ are not
stored; just sampled summary values.

\ejemplo{

\bigskip
\begin{example}
Consider the previous example query, now assuming that the labels of the grid
are weights associated with the points. We could associate with $A_{a,b}$ the
sum of its prefixes, $P_{a,b}[i] = \sum_{k=1}^i A_{a,b}[k]$. Then, instead of
counting the points in $[3,6] \times [2,6]$, we could sum up its weights:
since all the points are in $S_{2,3}[2,2]$ and $S_{4,5}[2,4]$, the sum of the
weights is $(P_{2,3}[2]-P_{2,3}[1])+(P_{4,5}[4]-P_{4,5}[1])$. We can then
compute sums in time $O(\log r)$, but use $O(p\log r\log N)$ bits. Sampling
can reduce this space while increasing the time.
\end{example}

}

In this paper we show that the $O(\log^3 r)$ time can be improved
to $O(\log^{2+\epsilon}r)$, for any constant $\epsilon>0$, within the same 
asymptotic space; see Theorem~\ref{thm:sum} in Section~\ref{sec:count}.

\subsection{Grammar compression} \label{sec:grammar}

Let $T[1,N]$ be a sequence of symbols over alphabet $[1,\sigma]$. Grammar
compressing $T$ means finding a context-free grammar that generates $T$ and
only $T$. The grammar can then be used as a substitute for $T$, which provides
good compression when $T$ is repetitive. We are interested, for simplicty, in
grammars in Chomsky normal form, where the rules are of the form 
$A \rightarrow BC$ or $A \rightarrow a$, where $A$, $B$, and $C$ are 
nonterminals and $a \in [1,\sigma]$ is a terminal symbol. For every grammar, 
there is a proportionally sized grammar in this form.

A Lempel-Ziv parse \cite{LZ76} of $T$ cuts $T$ into $z$ phrases, so that each
phrase $T[i,j]$ appears earlier in $T[i',j']$, with $i' < i$. It is known that
the smallest grammar generating $T$ must have at least $z$ rules
\cite{Ryt03,CLLPPSS05}, and that it is possible to convert a Lempel-Ziv parse 
into a grammar with $r=O(z \log(N/z))$ rules 
\cite{Ryt03,CLLPPSS05,Sak05,Jez15,Jez16}. Furthermore, such grammars can be 
balanced,
that is, the parse tree is of height $O(\log N)$. By storing the length of
the string to which every nonterminal expands, it is easy to access any
substring $T[i,j]$ from its compressed representation in time $O(j-i+\log N)$
by tracking down the range in the parse tree. This can be done even on 
unbalanced grammars \cite{BLRSRW15}. The total space of this 
representation, with a grammar of $r$ rules, is $O(r\log N)$ bits.

\ejemplo{

\begin{figure}[t]
\begin{center}
\includegraphics[width=0.7\textwidth]{grammar.pdf}
\end{center}
\caption{The grammar of three strings (or documents) $D_1$, $D_2$, and $D_3$,
each with an edit with respect to the previous one. On top we show the parse
tree of each document, in the middle the grammar in Chomsky normal form
(removing the rules $A \rightarrow a$ for conciseness), and in the bottom the
{\em inverted list} of the documents where each nonterminal appears.}
\label{fig:grammar}
\end{figure}

\bigskip
\begin{example}
Fig.~\ref{fig:grammar} shows the grammar compression of a set of three
strings, or documents, $D_1 = \texttt{abracada}$, $D_2 = \texttt{abrakada}$, 
and $D_3 = \texttt{ablakada}$. Each document has edits with respect to the
previous one. This follows the way in which we will use grammars in this 
article: $D_1$ has a trivial balanced grammar, and then the grammar of each
$D_d$ is equal to that of $D_{d-1}$ except for the edits, which induce new
terminals up to the root. Disregard the inverted lists for now.
\end{example}

}

\subsection{Grammar-based indexing} \label{sec:index}

The pattern-matching index of Claude and Navarro \cite{CNfi10} builds on a
grammar in Chomsky normal form that generates a text $T[1,N]$, with $r$ rules
of the form $A \rightarrow BC$. Let $s(A)$ be the string generated by 
nonterminal $A$. Then they collect the distinct strings $s(B)$ for all those 
nonterminals $B$, reverse them, and lexicographically sort them, obtaining
$s(B_1)^{rev}<\ldots<s(B_{r'})^{rev}$, for $r' \le r$. They also collect the
distinct 
strings $s(C)$ for all those nonterminals $C$ and lexicographically sort them,
obtaining $s(C_1)<\ldots<s(C_{r''})$, for $r'' \le r$. They create a set of 
points in $[1,r']\times[1,r'']$ so that $(i,j)$ is a point (corresponding to 
nonterminal $A$) if the rule that defines $A$ is $A \rightarrow B_i C_j$. Those
$r$ points are stored in a wavelet tree. Note that the nonterminals of the form
$A \rightarrow a$ are listed as some $B_i$ or some $C_j$ (or both), yet only 
the rules of the form $A \rightarrow BC$ have associated points in the grid. 
Since there may be many points per column, we use the coordinate mapping 
described in Section~\ref{sec:wt}. The space is thus $r\log r'' + o(r\log r'')
+O(r+r') \le r\log r + o(r\log r)$ bits.

\ejemplo{

\begin{figure}[t]
\begin{center}
\includegraphics[width=0.5\textwidth]{grid.pdf}
\end{center}
\caption{The grid structure formed by the sorted strings $s(B_i)^{rev}$ (on top)
and $s(C_j)$ (on the left). Points are labeled by the nonterminals $A$ that
connect $A \rightarrow B_i C_j$.}
\label{fig:grid}
\end{figure}

\bigskip
\begin{example}
Fig.~\ref{fig:grid} shows the grid corresponding to the grammar of
Fig.~\ref{fig:grammar}. If we separate points in the same column as explained, 
the result is the grid of Fig.~\ref{fig:wtree}, which is represented by that 
wavelet tree, and the labels correspond to the nonterminals. The mapping 
bitvector is $R=1111101011111$.
\end{example}

}

To search for a pattern $P[1,m]$, they first find the {\em primary occurrences},
that is, those that appear when $B$ is concatenated with $C$ in a rule
$A \rightarrow BC$. The {\em secondary occurrences}, which appear when $A$ is 
used elsewhere, are found in a way that does not matter for this paper.
To find the primary occurrences, they cut $P$ into two nonempty parts 
$P = P_1 P_2$, in the $m-1$ possible ways. For each cut, they binary search 
for $P_1^{rev}$ in the sorted set $s(B_1)^{rev}, \ldots, s(B_{r'})^{rev}$ and 
for $P_2$ in the sorted set $s(C_1), \ldots, s(C_{r''})$. Let $[x_1,x_2]$ be 
the interval obtained for $P_1^{rev}$ and $[y_1,y_2]$ the one obtained for 
$P_2$. Then all the points in $[x_1,x_2] \times [y_1,y_2]$, for all the $m-1$ 
partitions of $P$, are the primary occurrences. These are tracked down the
wavelet tree, where the label $A$ of the rule $A \rightarrow B_iC_j$ is 
explicitly stored at the leaf position of the point $(i,j)$. We then know that 
$P$ appears in $s(A)[|s(B_i)|-|P_1|+1,|s(B_i)|+|P_2|]$.

\ejemplo{

\bigskip
\begin{example}
To search for $P=\texttt{bra}$ we try two partitions, $P_1=\texttt{b}$ with
$P_2=\texttt{ra}$, and $P_1=\texttt{br}$ with $P_2=\texttt{a}$. In the first
partition, we find for $P_1$ the range $[\texttt{ba},\texttt{ba}]$ in the
strings $s(B_i)^{rev}$, which after mapping it to the grid of 
Fig.~\ref{fig:wtree} becomes $[x_1,x_2]=[7,8]$. For $P_2$ we find the range 
$[\texttt{ra},\texttt{ra}]$ in the strings $s(C_j)$, which is
$[y_1,y_2]=[7,7]$ in Fig.~\ref{fig:wtree}. The wavelet tree search
then ends up in the rightmost leaf. If we explicitly store the strings
$A_{a,b}$ at the leaves, we may recover the nonterminal $\mathsf{E}$, which is
the only occurrence of $P$ under this partition $P_1 P_2$. The second
partition does not return results, thus nonterminal $\mathsf{E}$ contains the
only primary occurrence of $P$.
\end{example}

}

The special case $m=1$ is handled by binary searching the $(B_i)^{rev}$s or the
$(C_j)$s for the only nonterminal $A \rightarrow P[1]$. This, if exists, is 
the only primary occurrence of $P$.

To search for $P_1^{rev}$ or for $P_2$, the grammar is used to extract the
required substrings of $T$ in time $O(m + \log N)$, so the overall search time 
to find the $\nocc$ nonterminals containing the primary occurrences 
is $O(m\log r (m+\log N) + \log r \cdot 
\nocc)$. Let us describe the fastest known variant that uses $O(r\log N)$
bits, disregarding constant factors in the space. Within $O(r\log N)$ bits,
one can store Patricia trees \cite{Mor68} on the strings $s(B_i)^{rev}$
and $s(C_j)$, to speed up binary searches and reduce the time to
$O(m (m+\log N) + \log r \cdot nocc)$. Also, one can use the structure of
Gasieniec et al.\ \cite{GKPS05} that, within $O(r\log N)$ further bits, allows
extracting any prefix/suffix of any nonterminal in constant time per symbol
(see Claude and Navarro \cite{CN12} for more details). Since in our search we
only access prefixes/suffixes of whole nonterminals, this further reduces the
time to $O(m^2 + (m+nocc)\log r)$. Finally, we can use the technique for
faster reporting described in Section~\ref{sec:wt} to obtain time
$O(m^2 + m\log r + \log^\epsilon r \cdot nocc)$, for any constant $\epsilon>0$.

\paragraph*{Faster locating on balanced grammars}
If the grammar is balanced, however, we can do better within $O(r\log N)$
bits using the most recent developments. We can store z-fast tries 
\cite[App.~H.3]{BBPV18} on the sets $s(B_1)^{rev}, \ldots, s(B_{r'})^{rev}$ and 
$s(C_1), \ldots, s(C_{r''})$. We can also associate with each nonterminal $A$ a 
Karp-Rabin fingerprint \cite{KR87} for $s(A)$. If the balanced grammar is in 
Chomsky
normal form, then any substring of $T$ is covered by $O(\log N)$ maximal 
nonterminals, so its fingerprint can be assembled in time $O(\log N)$. 
Otherwise, we can convert it into Chomsky normal form while perserving its
asymptotic size and balancedness.%
\footnote{To convert a rule $A \rightarrow B_1\ldots B_t$ to Chomsky normal
form, instead of building a balanced binary tree of $t-1$ intermediate rules, 
use the tree corresponding to the Shannon codes \cite[Sec.~5.4]{CT06} of the 
probabilities $|s(B_i)|/|s(A)|$. Those guarantee that the leaf for each
$B_i$ is at depth $\left\lceil \lg \frac{|s(A)|}{|s(B_i)|}\right\rceil$. Any
root-to-leaf path of length $h$, when expanded by this process, telescopes to
$h+\lg N$.}
It is possible to build the fingerprints so as to ensure no 
collisions between substrings of $T$ \cite{GGKNP14}. We can also extract any
substring of length $m$ of $T$ in time $O(m+\log N)$, and even in time 
$O(m)$ if they are prefixes or suffixes of some $s(A)$ \cite{GKPS05}. With all 
those elements, we can build a scheme \cite[Lem.~5.2]{GNP18} that can find the
lexicographic ranges of 
the $m-1$ prefixes $P_1^{rev}$ in $s(B_1)^{rev}, \ldots, s(B_{r'})^{rev}$ and 
the $m-1$ suffixes $P_2$ in $s(C_1), \ldots, s(C_{r''})$, all in time
$O(m\log N)$. This reduces the time obtained in the preceding paragraph to
$O(m\log N + \log^\epsilon r \cdot \nocc)$, for any constant $\epsilon>0$. 
We will use this result for document listing, but it is of independent interest
as a grammar-based pattern-matching index. Note we have disregarded the
secondary occurrences, but those are found in $O(\log\log r)$ time each with
structures using $O(r\log N)$ bits \cite{CN12}.

\bigskip
\begin{theorem} \label{thm:locate}
Let text $T[1,N]$ be represented by a balanced grammar of size $r$. Then there 
is an index of $O(r\log N)$ bits that locates the $\nocc$ occurrences in 
$T$ of a pattern $P[1,m]$ in time $O(m\log N+\log^\epsilon r \cdot \nocc)$, 
for any constant $\epsilon>0$.
\end{theorem}

\paragraph*{Counting}
This index locates the occurrences of $P$ one by one, but cannot count them
without locating them all. This is a feature easily supported by 
suffix-array-based compressed indexes \cite{MNSV10,GNP18} in $O(m\log N)$ time 
or less, but so far unavailable in grammar-based or Lempel-Ziv-based compressed 
indexes.
In Theorem~\ref{thm:count} of Section~\ref{sec:count} we offer for the first
time efficient counting for grammar-based indexes. Within their same asymptotic
space, we can count in time $O(m^2 + m\log^{2+\epsilon} r)$ for any constant
$\epsilon>0$ (and $O(m(\log N+\log^{2+\epsilon} r))$ if the grammar is
balanced). For a text parsed into $z$ Lempel-Ziv phrases we obtain,
in Theorem~\ref{thm:countz}, 
$O(m\log^{2+\epsilon} (z\log(N/z)))$ time and $O(z\log(N/z)\log N)$ bits.

\paragraph*{Document listing}
The original structure was also unable to perform document listing without
locating all the occurrences and determining the document where each belongs.
Claude and Munro \cite{CM13} showed how to extend it in order to support 
document listing on a collection $\mathcal{D}$ of $D$ string documents, which 
are concatenated into a text $T[1,N]$. A grammar is built on $T$, where 
nonterminals are not allowed to cross document boundaries. To each nonterminal 
$A$ they associate the increasing list
$\ell(A)$ of the identifiers of the documents (integers in $[1,D]$) where $A$ 
appears. To perform document listing, they find all the primary occurrences 
$A \rightarrow BC$ of all the partitions of $P$, and merge their lists.%
\footnote{In pattern matching, the same nonterminal $A$ may be found several
times with different partitions $P=P_1 P_2$, and these yield different
occurrences. For document listing, however, we are only interested in the
nonterminal $A$.}
There is no useful worst-case time bound for this operation other than 
$O(r \cdot \ndoc)$.
To reduce space, they also grammar-compress the sequence of all
the $r$ lists $\ell(A)$. They give no worst-case space bound for the 
compressed lists (other than $O(rD\log D)$ bits). 

\ejemplo{

\bigskip
\begin{example}
The bottom of Fig.~\ref{fig:grammar} shows the lists $\ell(\cdot)$ associated
with the grammar nonterminals. After our search for $P=\texttt{bra}$ in the
previous example, which returned only the nonterminal $\mathsf{E}$, we may 
directly report $\ell(\mathsf{E})=1,2$, the list of all the documents where
$P$ appears.
\end{example}

}

At the end of Section~\ref{sec:space} we show that, under our repetitiveness
model, this index can be tweaked to occupy 
$O(n\log N + s\log^2 N)$ bits, close to what
can be expected from a grammar-based index according to our discussion. Still,
it gives no worst-case guarantees for the document listing time. In 
Theorem~\ref{thm:main} we show that, by multiplying the
space by an $O(\log D)$ factor, document listing is possible in time
$O(m\log^{1+\epsilon} N \cdot \ndoc)$ for any constant $\epsilon>0$.

\section{Our Document Listing Index} 

We build on the basic structure of Claude and Munro \cite{CM13}. 
Our main idea is to take advantage of the fact that the $\nocc$ primary 
occurrences to detect in Section~\ref{sec:index} are found as points in
the two-dimensional structure, along $O(\log r)$ ranges within wavelet tree
nodes (recall Section~\ref{sec:wt}) for each partition of $P$. Instead of 
retrieving the $\nocc$ individual lists, decompressing and merging them
\cite{CM13}, we will use the techniques to extract the distinct elements of 
a range seen in Section~\ref{sec:listing}. This will drastically reduce the
amount of merging necessary, and will provide useful upper bounds on the 
document listing time.

\subsection{Structure}

We store the grammar of $T$ in a way that it allows direct access for pattern 
searches, as well as the wavelet tree for the points $(i,j)$ of $A \rightarrow
B_i C_j$, the Patricia 
trees, and extraction of prefixes/suffixes of nonterminals, all in 
$O(r\log N)$ bits; recall Section~\ref{sec:index}.

Consider any sequence $S_{a,b}[1,q]$ at a wavelet tree node handling the range 
$[a,b]$ (recall that those sequences are not explicitly stored). Each element 
$S_{a,b}[k]=j$ corresponds to a point $(i,j)$ associated with a nonterminal 
$A_k \rightarrow B_iC_j$. Consider the sequence of associated labels
$A_{a,b}[1,q] = A_1,\ldots,A_q$ (not explicitly stored either).
Then let $L_{a,b} = \ell(A_1) \cdot \ell(A_2) \cdots
\ell (A_q)$ be the concatenation of the inverted lists associated with the
nonterminals of $A_{a,b}$, and let $M_{a,b}=1 0^{|\ell(A_1)|-1} 
1 0^{|\ell(A_2)|-1} \ldots 1 0^{|\ell(A_q)|-1}$ mark where each list begins
in $L_{a,b}$. Now let $E_{a,b}$ be the $E$-array corresponding to
$L_{a,b}$, as described in Section~\ref{sec:listing}. As in that section, we do
not store $L_{a,b}$ nor $E_{a,b}$, but just the RMQ structure on $E_{a,b}$, 
which together with $M_{a,b}$ will be used to retrieve the unique documents 
in a range $S_{a,b}[i,j]$. 

Since $M_{a,b}$ has only $r$ 1s out of (at most) $rD$ bits across all the 
wavelet tree nodes of the same level, it can be stored with $O(r\log D)$ bits 
per level \cite{OS07}, and $O(r\log r\log D)$
bits overall. On the other hand, as we will show, $E_{a,b}$ is formed by a few 
increasing runs, say $\rho$ across the wavelet tree nodes of the same level, 
and therefore we represent its RMQ structure using the technique of 
Section~\ref{sec:rmq}. The total space used by those RMQ structures is then 
$O(\rho\log r\log (rD/\rho))$ bits.

Finally, we store the explicit lists $\ell(A_k)$ aligned to the sequences
$A_{j,j}$ of the wavelet tree leaves $j$, so that the list of any element 
$A_{a,b}[k]$ is reached
in $O(\log r)$ time by tracking down the element. Those lists, of maximum total
length $rD$, are grammar-compressed as well, just as in the basic scheme
\cite{CM13}. If the grammar has $l$ rules, then the total compressed
size is $O(l\log(rD))$ bits to allow for direct access in $O(\log(rD))$ time,
see Section~\ref{sec:grammar}.

Our complete structure uses 
$O(r\log N + r\log r\log D +\rho\log r\log (rD/\rho) + l\log(rD))$ bits.

\ejemplo{

\begin{figure}[t]
\begin{center}
\includegraphics[width=0.5\textwidth]{node.pdf}
\end{center}
\caption{Our structure on a wavelet tree node. Only those in black are actually
stored.}
\label{fig:node}
\end{figure}

\bigskip
\begin{example}
Fig.~\ref{fig:node} shows our structure on a wavelet tree node, specifically
the left child of the root in Fig.~\ref{fig:wtree}. Only the bitvectors
$B_{a,b}$, $M_{a,b}$, and $F_{a,b}$ are stored explicitly, plus the RMQ 
structures on $E'_{a,b}$ (the headers of the nondecreasing runs in $E_{a,b}$). 
Bitvectors $B_{a,b}$ (plus the all structures associated with the grammar of 
$T$) contribute $O(r\log N)$ bits, bitvectors $M_{a,b}$ contribute 
$O(r\log r\log D)$ bits, and bitvectors $F_{a,b}$ and the RMQ structures on
$E'_{a,b}$ contribute $O(\rho \log r \log(rD/\rho))$ bits.
On the leaf nodes, the sequences $L_{a,b}$ are explicitly stored as well 
(in grammar-compressed form), contributing $O(l\log(rD))$ further bits.
\end{example}

}
\subsection{Document listing} \label{sec:dlist}

A document listing query proceeds as follows. We cut $P$ in the $m-1$ possible 
ways, and for each way identify the $O(\log r)$ wavelet tree nodes and ranges  
$A_{a,b}[i,j]$ where the desired nonterminals lie. Overall, we have 
$O(m\log r)$ ranges and need to take the union of the inverted lists of all 
the nonterminals in those ranges. We extract the distinct documents in each 
corresponding range $L_{a,b}[i',j']$ and then compute their union. If a range 
has only one element, we can simply track it to the leaves, where its list 
$\ell(A_k)$ is stored, and decompress the whole list. Otherwise, we use a more
sophisticated mechanism.

\ejemplo{

\bigskip
\begin{example} 
To have an interesting example, consider we search for $P=[\texttt{b-r}]\,
[\texttt{a-l}]$, where we admit ranges of symbols, in the index of 
Fig.~\ref{fig:grid}. Then we partition it into $P_1 = [\texttt{b-r}]$ and 
$P_2=[\texttt{a-l}]$. The searches, after mapping to the grid of 
Fig.~\ref{fig:wtree}, yield the range $[x_1,x_2]=[7,13]$ and $[y_1,y_2]=[1,6]$.
This is mapped to the wavelet tree nodes and ranges 
$A_{1,3}[3,7]$, $A_{4,5}[1,0]$ (empty
range), and $A_{6,6}[1,1]$. As seen in Fig.~\ref{fig:node}, the interval
of $A_{1,3}[3,7]$ corresponds to the concatenation of lists $L_{1,3}[5,13]=
1,1,2,3,2,3,3,1,2$ (obtained with rank and select on $M_{1,3}$). We will obtain
the distinct documents $1,2,3$ in this range. The range $A_{6,6}[1,1]$ contains
only the list $\ell(\mathsf{E''})=3$ (see Fig.~\ref{fig:grammar}). Thus we
find document $3$ twice, in different wavelet tree nodes. On longer patterns,
we can also find them repeated across different partitions $P_1 P_2$.
\end{example}

}

We use in principle the document listing technique of 
Section~\ref{sec:listing}. Let $A_{a,b}[i,j]$ be a range from where to
obtain the distinct documents. We compute $i'=\select_1(M_{a,b},i)$ and 
$j'=\select_1(M_{a,b},j+1)-1$, and obtain the distinct elements in 
$L_{a,b}[i',j']$, by using RMQs on $E_{a,b}[i',j']$. Recall that, as in
Section~\ref{sec:rmq}, we use a run-length compressed RMQ structure on 
$E_{a,b}$. With this arrangement, every RMQ operation takes time 
$O(\log\log(rD))$ plus the time to accesses two cells in $E_{a,b}$.
Those accesses are made to compare a run 
head with the leftmost element of the query interval, $E_{a,b}[i']$.
The problem is that we have not represented the cells of $E_{a,b}$, nor we can
easily compute them on the fly. 

Barbay et al.\ \cite[Thm.~3]{BFN12} give a representation that
determines the position of the minimum in $E_{a,b}[i',j']$ without the need to
perform the two accesses on $E_{a,b}$. They need 
$\rho\log(rD)+\rho\log(rD/\rho)+O(\rho)+o(rD)$ bits. The last term is,
unfortunately, is too high for us%
\footnote{Even if we get rid of the $o(rD)$ component, the $\rho\log(rD)$
term becomes $O(s\log^3 N)$ in the final space, which is larger than
what we manage to obtain. Also, using it does not make our solution faster.}.

Instead, we modify the way the distinct elements are obtained, so that
comparing the two cells of $E_{a,b}$ is unnecessary. In the same
spirit of Sadakane's solution (see Section~\ref{sec:listing}) we use a
bitvector $V[1,D]$ where we mark the documents already reported. Given a range
$A_{a,b}[i,j]$ (i.e., $L_{a,b}[i',j']=\ell(A_{a,b}[i])\cdots\ell(A_{a,b}[j])$),
we first track $A_{a,b}[i]$ down the wavelet tree, recover and decompress its 
list $\ell(A_{a,b}[i])$, and mark all of its documents in $V$. Note that all 
the documents in a list $\ell(\cdot)$ are different.  Now we do the same with 
$A_{a,b}[i+1]$, decompressing $\ell(A_{a,b}[i+1])$ left to right and marking 
the documents in $V$, and so on, until we decompress a document 
$\ell(A_{a,b}[i+d])[k]$ that is already marked in $V$. Only now we use the RMQ 
technique of Section~\ref{sec:rmq} on the interval $E_{a,b}[x,j']$, where 
$x=\select_1(M_{a,b},i+d)-1+k$, to obtain the next document to report. This
technique, as explained, yields two candidates: one is $E_{a,b}[x]$, where
$L_{a,b}[x] = \ell(A_{a,b}[i+d])[k]$ itself, and the other is some run head 
$E_{a,b}[k']$, where we can obtain $L_{a,b}[k']$ from the wavelet tree leaf 
(i.e., at $\ell(A_{a,b}[t])[u]$, where $t = \rank_1(M_{a,b},k')$ and 
$u=k'-\select_1(M,t)+1$).
But we know that $L_{a,b}[x]$ was already found twice and thus $E_{a,b}[x] \ge 
i'$, so we act as if the RMQ was always $E_{a,b}[k']$: If the correct RMQ 
answer was $E_{a,b}[x]$ then, since $i' \le E_{a,b}[x] \le E_{a,b}[k']$, we 
have that $L_{a,b}[k']$ is already reported and we will stop anyway. Hence, if 
$L_{a,b}[k']$ is already reported we stop, and otherwise we report it and 
continue recursively on the intervals $E_{a,b}[i',k'-1]$ and $E_{a,b}[k'+1,j']$.
On the first, we can continue directly, as we still know that $L_{a,b}[i']$ was
found twice. On the second interval, instead, we must restore the invariant 
that the leftmost element was found twice. So we find out with $M$ the list and
position of $L_{a,b}[k'+1]$, and traverse the list from that position onwards,
reporting documents until finding one that had already been reported.

If the RMQ algorithm does not return any second candidate $E_{a,b}[k']$ (which
happens when there are no run heads in $E_{a,b}[i'+1,j']$) we can simply stop,
since the minimum is $E_{a,b}[i']$ and $L_{a,b}[i']$ is already reported.
The correctness of this document listing algorithm is formally proved in
Appendix~\ref{sec:app}.

The $m-1$ searches for partitions of $P$ take time $O(m^2)$, as seen in
Section~\ref{sec:index}.
In the worst case, extracting each distinct document in the range requires an 
RMQ computation without access to $E_{a,b}$ ($O(\log\log(rD))$ time), tracking
an element down the wavelet tree ($O(\log r)$ time), and extracting an element 
from its grammar-compressed list $\ell(\cdot)$ ($O(\log(rD)$ time).
This adds up to $O(\log(rD))$ time per document 
extracted in a range.
In the worst case, however, the same documents are extracted over 
and over in all the $O(m\log r)$ ranges, and therefore the final search 
time is $O(m^2 + m\log r \log (rD) \cdot \ndoc)$.

\section{Analysis in a Repetitive Scenario}

Our structure uses 
$O(r\log N + r\log r\log D +\rho\log r\log (rD/\rho) + l\log(rD))$ 
bits, and performs document listing in time 
$O(m^2 + m\log r\log (rD)\cdot\ndoc)$.
We now specialize those formulas under our repetitiveness model. 
Note that our index works on any string collection; we use
the simplified model of the $D-1$ copies of a single document of length $n$,
plus the $s$ edits, to obtain analytical results that are easy to interpret in 
terms of repetitiveness. 

We also assume a particular strategy to generate the grammars in order
to show that it 
is possible to obtain the complexities we give. This involves determining the 
minimum number of edits that distinguishes each document from the previous one.
If the $s$ edit positions are not given 
explicitly, the optimal set of $s$ edits can still be obtained at construction 
time, with cost $O(Ns)$, using dynamic programming \cite{Ukk85}. 

\subsection{Space} \label{sec:space}

Consider the model where we have $s$ single-character edits affecting a {\em 
range} of document identifiers. This includes the model where each edit affects
a single document, as a special case.
The model where the documents form a tree of versions, and each edit affects 
a whole subtree, also boils down to the model of ranges by numbering the 
documents according to their preorder position in the tree of versions.

An edit that affects a range of documents $d_i,\ldots,d_j$ will be regarded as
two edits: one that applies the change at $d_i$ and one that undoes it at
$d_j$ (if needed, since the edit may be overriden by another later edit). Thus,
we will assume that there are at most $2s$ edits, each of which affects
all the documents starting from the one where it applies. We will then assume
$s \ge (D-1)/2$, since otherwise there will be identical documents, 
and this is easily reduced to a smaller collection with multiple identifiers 
per document. 

\paragraph*{Our grammar}

The documents are concatenated into a single text $T[1,N]$, where 
$N \le D(n+s)$. Our grammar for $T$ will be built over an alphabet of 
$O(N^{1/3})$ ``metasymbols'', which include all the possible strings of length
up to $\frac{1}{3}\log_\sigma N$. The first document is parsed into
$\lceil n/\frac{1}{3}\log_\sigma N\rceil$ metasymbols, on top of which we build
a perfectly balanced binary parse tree of height $h=\Theta(\log n)$ (for 
simplicity; any balanced grammar would do). All the 
internal nodes of this tree are distinct nonterminal symbols (unless they 
generate the same strings), and end up in a root symbol $S_1$.

Now we regard the subsequent documents one by one. For each new document $d$,
we start
by copying the parse tree from the previous one, $d-1$, including the start
symbol $S_d = S_{d-1}$. Then, we apply the edits that start at that document.
Let $h$ be the height of its parse tree. A character substitution requires 
replacing the metasymbol covering the position where the edit applies, and then
renaming the nonterminals $A_1,\ldots,A_h=S_d$ in the path from the parent of 
the metasymbol to the root. Each $A_i$ in the path is replaced by a new 
nonterminal $A_i'$ (but we reuse existing nonterminals to avoid duplicated
rules $A \rightarrow BC$ and $A' \rightarrow BC$). The nonterminals that do 
not belong to the path are not affected. A deletion proceeds similarly: we 
replace the metasymbol of length $k$ by one of length $k-1$ (for simplicity, 
we leave the metasymbol of length $0$, the
empty string, unchanged if it appears as a result of deletions). Finally, an
insertion into a metasymbol of length $k$ replaces it by one of length $k+1$,
unless $k$ was already the maximum metasymbol length, 
$\frac{1}{3}\log_\sigma N$. In this case we replace the metasymbol leaf by
an internal node with two leaves, which are metasymbols of length around
$\frac{1}{6}\log_\sigma N$. To maintain a balanced tree, we use the AVL 
insertion
mechanism, which may modify $O(h)$ nodes toward the root. This ensures that,
even in documents receiving $s$ insertions, the height of the parse tree will
be $O(\log(n+s))$.

The Chomsky normal form requires that we create nonterminals $A \rightarrow a$
for each metasymbol $a$ (which is treated as a single symbol); the first 
document creates $O(n/\log_\sigma N)$ nonterminals; and each edit
creates $O(\log(n+s))$ new nonterminals. Therefore, the final grammar size
is $r = \Theta(N^{1/3} + n/\log_\sigma N + s \log (n+s)) =
\Theta(n/\log_\sigma N + s\log N)$, where we used that either $n$ or $s$ is 
$\Omega(\sqrt{N})$ because $N \le D(n+s) \le (2s+1)(n+s)$.
Once all the edits are applied, we add a balanced tree on top of the $D$
symbols $S_d$, which asymptotically does not change $r$ (we may also avoid this
final tree and access the documents individually, since our accesses never cross
document borders).
Further, note that, since this grammar is balanced, Theorem~\ref{thm:locate}
allows us reduce its $O(m^2)$ term in the search time to $O(m\log N)$.

\paragraph*{Inverted lists}

Our model makes it particularly easy to bound $l$. 
Instead of grammar-compressing the lists, we store for each nonterminal a
plain inverted list encoded as a sequence of ranges of documents, as follows.
Initially, all the nonterminals that appear in the first document have a list 
formed by the single range $[1,D]$. Now we consider the documents $d$ one by 
one, with the invariant that a nonterminal appears in document $d-1$ iff the
last range of its list is of the form $[d',D]$. For each nonterminal that 
disappears in document $d$ (i.e., an edit removes its last occurrence), we 
replace the last range $[d',D]$ of its list by $[d',d-1]$. For each
nonterminal that (re)appears in document $d$, we add a new range $[d,D]$ to
its list. Overall, the total size of the inverted lists of all the
nonterminals is $O(r + s\log N)$, and each entry requires $O(\log D)$ bits.
Any element of the list is accessed with a predecessor query in $O(\log\log
D)$ time, faster than on the general scheme we described.

The use of metasymbols requires a special solution for patterns of length up to 
$\frac{1}{3}\log_\sigma N$, since some of their occurrences might not be found
crossing nonterminals. For all the $O(N^{1/3})$ possible patterns of up to that
length, we store the document listing answers explicitly, as inverted lists
encoding ranges of documents. These are created as for the nonterminals.
Initially, all the 
metasymbols that appear in the first document have a list formed by the single 
range $[1,D]$, whereas the others have an empty list. Now we consider the 
documents one by one. For each edit applied in document $d$, we consider each 
of the $O(\log^2_\sigma N)$ metasymbols of all possible lengths 
that the edit destroys. If this was the only occurrence of the metasymbol in
the document, we replace the last range $[d',D]$ of the list of the metasymbol
by $[d',d-1]$. Similarly, for each of the $O(\log^2_\sigma N)$ metasymbols of 
all possible lengths that the edit creates, if the metasymbol was not present
in the document, we add a new range $[d,D]$ to the list of the metasymbol.
Overall, the total size of the inverted lists of the metasymbols is
$O(N^{1/3} + s\log^2_\sigma N) \subseteq O(n+s\log^2_\sigma N)$, and each 
entry requires $O(\log D)$ bits.

\paragraph*{Run-length compressed arrays $E_{a,b}$}

Let us now bound $\rho$. When we have only the initial document, all 
the existing nonterminals mention document $1$, and thus $E=E_{a,b}$ has a 
single 
nondecreasing run. Now consider the moment where we include document $d$. We
will insert the value $d$ at the end of the lists of all the nonterminals $A$ 
that appear in document $d$. As long as document $d$ uses the same parse tree 
of document $d-1$, no new runs are created in $E$.

\bigskip
\begin{lemma}
If document $d$ uses the same nonterminals as document $d-1$, inserting it 
in the inverted lists does not create any new run in the $E$ arrays.
\end{lemma}
\begin{proof}
The positions $p_1,\ldots,p_k$ where we insert the document $d$ in the lists
of the nonterminals that appear in it, will be 
chained in a list where $E[p_{i+1}] = p_i$ and $E[p_1]=0$. Since all the 
nonterminals $A$ also appear in document $d-1$, the lists
will contain the value $d-1$ at positions $p_1-1,\ldots,p_k-1$, and we will
have $E[p_{i+1}-1]=p_i-1$ and $E[p_1-1]=0$. Therefore, the new values we
insert for $d$ will not create new runs: $E[p_1]=E[p_1-1]=0$ does not create a 
run, and neither can $E[p_{i+1}]=E[p_{i+1}-1]+1$, because if $E[p_{i+1}+1] < 
E[p_{i+1}]=p_i$, then we are only creating a new run if $E[p_{i+1}+1] = p_i-1$, 
but this cannot be since $E[p_{i+1}-1]=p_i-1=E[p_{i+1}+1]$ and in this case 
$E[p_{i+1}+1]$ should have pointed to $p_{i+1}-1$. 
\end{proof}

Now, each edit we apply on $d$ makes $O(\log N)$ 
nonterminals appear or disappear, and thus $O(\log N)$ values of $d$ appear
or disappear in $E$. Each such change may break a run. Therefore, $E$ may
have at most $\rho = O(s\log N)$ runs per wavelet tree level (all the lists
appear once in each level, in different orders).

\paragraph*{Total}

The total size of the index can then be expressed as follows. The 
$O(r\log r\log D)$ bits coming from the sparse bitvectors $M$, is
$O(r\log N\log D)$ (since $\log r = \Theta(\log(ns)) = \Theta(\log N)$), and 
thus it is $O(n\log\sigma \log D + s\lg^2 N\lg D)$. This subsumes the 
$O(r\log N)$ bits of the grammar and the wavelet tree. The inverted lists can 
be represented with $O((r+s\log N)\log D)$ bits, and the explicit answers for 
all the metasymbols require $O((n+s\log^2_\sigma N)\log D)$ bits. Finally, the
$O(\rho\log r\log(rD/\rho))$ bits of the structures $E$ are monotonically 
increasing with $\rho$, so since $\rho = O(s\log N) = O(r)$, we can upper 
bound it by replacing $\rho$ with $r$, obtaining $O(r\log r\log D)$
as in the space for $M$. Overall, the structures add up to
$O((n\log\sigma + s\log^2 N)\log D)$ bits. 

Note that we can also analyze the space required by Claude and Munro's 
structure \cite{CM13}. They only need the $O(r\log N)$ bits of the grammar
and the wavelet tree, which avoiding the use of metasymbols is
$O(n\log N+s\log^2 N)$ bits. Although smaller than ours almost by an 
$O(\log D)$ factor, their search time has no useful bounds.

\subsection{Time}

If $P$ does not appear in $\mathcal{D}$, we note it in time $O(m\log N)$, since
all the ranges are empty of points. Otherwise, our search time is
$O(m\log N + m\log r\log(rD) \cdot \ndoc) = O(m\log^2 N \cdot \ndoc)$.
The $O(\log(rD))$ cost corresponds to accessing a list $\ell(A)$ from the
wavelet tree, and includes the $O(\log r)$ time to reach the leaf and the
$O(\log D)$ time to access a position in the grammar-compressed list. Since we
have replaced the grammar-compressed lists by a sequence of ranges, this last
cost is now just $O(\log\log D) \subseteq O(\log\log r)$. 
As seen in Section~\ref{sec:wt}, it is possible to 
reduce the $O(\log r)$ tracking time to $O((1/\epsilon)\log^\epsilon r)$ 
for any $\epsilon>0$, within $O((1/\epsilon)r\log N)$ bits. In this case,
the lists $\ell(A)$ are associated with the symbols at the root of the wavelet
tree, not the leaves.

\bigskip
\begin{theorem} \label{thm:main}
Let collection $\mathcal{D}$, of total size $N$, be formed by an initial 
document of length $n$ plus $D-1$ copies of it, with $s$ single-character
edit operations performed on ranges or subtrees of copies. Then 
$\mathcal{D}$ can be represented within $O((n\log\sigma +s\log^2 N)\log D)$ 
bits, so that the $\ndoc>0$ documents where a pattern of length $m$ appears can 
be listed in time $O(m\log^{1+\epsilon} N \cdot \ndoc)$, for
any constant $\epsilon>0$. If the pattern does not appear in $\mathcal{D}$, we
determine this is the case in time $O(m\log N)$.
\end{theorem}

We can also obtain other tradeoffs. For example, with $\epsilon = 1/\lg\lg r$
we obtain $O((n\log\sigma +s\log^2 N)(\log D+\lg\lg N))$ bits of space and
$O(m\log N \log\log N \cdot \ndoc)$ search time.

\section{Counting Pattern Occurrences} \label{sec:count}

Our idea of associating augmented information 
with the wavelet tree of the grammar has
independent interest. We illustrate this by developing a variant where
we can count the number of times a pattern $P$ occurs in the text without
having to enumerate all the occurrences, as is the case with all the 
grammar-based indexes \cite{CNfi10,CN12,CE18}. In these structures, the primary
ocurrences are found as points in various ranges of a grid (recall
Section~\ref{sec:index}). Each primary
occurrence then triggers a number of secondary occurrences, disjoint from those
triggered by other primary occurrences. These secondary occurrences depend only
on the point: if $P$ occurs when $B$ and $C$ are concatenated in the rule
$A \rightarrow BC$, then every other occurrence of $A$ or of its ancestors
in the parse tree produces a distinct secondary occurrence. Even if the same
rule $A \rightarrow BC$ is found again for another partition $P=P_1P_2$, the 
occurrences are different because they have different offsets inside $s(A)$.

We can therefore associate with each point the number of secondary occurrences
it produces, and thus the total number of occurrences of $P$ is the sum of the
numbers associated with the points contained in all the ranges. By 
augmenting the wavelet tree (recall Section~\ref{sec:wt}) of the grid,
the sum in each range can be computed in time $O(\log^3 r)$, using $O(r\log N)$
further bits of space for the grid \cite[Thm.~6]{NNR13}.%
\footnote{Although the theorem states that it must be
$t \ge 1$, it turns out that one can use $t=\log r / \log N$ (i.e., $\tau =
\log r$) to obtain this tradeoff (our $r$ is their $n$ and our $N$ is their
$W$).} We now show how this result can be improved to time $O(\log^{2+\epsilon}
r)$ for any constant $\epsilon>0$. Instead of only sums, we consider the more 
general case of a finite group \cite{NNR13}, so our particular case is 
$([0,N],+,-,0)$.

\bigskip
\begin{theorem} \label{thm:sum}
Let a grid of size $r \times r$ store $r$ points with associated values
in a group $(G,\oplus,^{-1},0)$ of $N=|G|$ elements. For any $\epsilon>0$, a
structure of $O((1/\epsilon)r\log N)$ bits can compute the sum 
$\oplus$ of the values in any rectangular range in time 
$O((1/\epsilon)\log^{2+\epsilon} r)$.
\end{theorem}
\begin{proof}
We modify the proof Navarro et al.~\cite[Thm.~6]{NNR13}.
They consider, for the sequence $S_{a,b}$ of each wavelet tree node, the 
sequence of associated values $A_{a,b}$. They store a cumulative array
$P_{a,b}[0] = 0$ and $P_{a,b}[i+1] = P_{a,b}[i] \oplus A_{a,b}[i+1]$, so
that any range sum $\oplus_{i \le k \le j} A_{a,b}[k] = P_{a,b}[j] \oplus 
P_{a,b}[i-1]^{-1}$ is computed in constant time. The space to store $P_{a,b}$ 
across all the levels is $O(r \log r \log N)$ bits. To reduce it to
$O(r\log N)$, they store instead the cumulative sums of a sampled array 
$A'_{a,b}$, where $A'_{a,b}[i] = \oplus_{(i-1)\log r<k \le i\log r} A_{a,b}[k]$.
They can then compute any range sum over $A'_{a,b}$, with which they can 
compute any range sum over $A_{a,b}$ except for up to $\log r$ elements in 
each extreme. Each of 
those extreme elements can be tracked up to the root in time $O((1/\epsilon)
\log^\epsilon r)$, for any $\epsilon>0$, using $O((1/\epsilon)r\log r)$ bits,
as described at the end of Section~\ref{sec:wt}. The root sequence $A_{1,r}$ 
is stored explicitly, in $r\log N$ bits. Therefore, we 
can sum the values in any range of any wavelet tree node in time 
$O((1/\epsilon)\log^{1+\epsilon} r)$. Since any two-dimensional range is 
decomposed into $O(\log r)$ wavelet tree ranges, we can find the sum in time 
$O((1/\epsilon)\log^{2+\epsilon} r)$.
\end{proof}

This immediately yields the first grammar-compressed 
index able to count pattern occurrences without locating them one by one.

\bigskip
\begin{theorem} \label{thm:count}
Let text $T[1,N]$ be represented by a grammar of size $r$. Then there exists
an index of $O(r\log N)$ bits that can count the number of occurrences of
a pattern $P[1,m]$ in $T$ in time $O(m^2+m\log^{2+\epsilon} r)$, for any
constant $\epsilon>0$. If the grammar is balanced, the time can be made
$O(m(\log N+\log^{2+\epsilon} r))$.
\end{theorem}

Further, since we can produce a balanced grammar of size $r = O(z\log(N/z))$
for a text of length $N$ with a Lempel-Ziv parse of size $z$
\cite{Ryt03,CLLPPSS05,Sak05,Jez15,Jez16}, we also obtain a data structure
whose size is bounded by $z$.

\bigskip
\begin{theorem} \label{thm:countz}
Let text $T[1,N]$ be parsed into $z$ Lempel-Ziv phrases. Then there exists
an index of $O(z\log(N/z)\log N)$ bits that can count the number of occurrences
of a pattern $P[1,m]$ in $T$ in time
$O(m\log N+m\log^{2+\epsilon} (z\log(N/z))) = O(m\log^{2+\epsilon} N)$, for any
constant $\epsilon>0$.
\end{theorem}

%

\section{Conclusions} \label{sec:conclusions}

We have presented the first document listing index with worst-case space and
time guarantees that are useful for repetitive collections. On a collection of
size $N$ formed by an initial document of length $n$ and $D-1$ copies it, with
$s$ single-character edits applied on individual documents, or ranges 
of documents (when there is a linear structure of versions), or subtrees of 
documents (when there is a hierarchical structure of versions),
our index uses $O((n\log\sigma +s\log^2 N)\log D)$ bits and lists the 
$\ndoc>0$ documents where a pattern of length $m$ appears in time 
$O(m \log^{1+\epsilon} N \cdot \ndoc)$, for any constant
$\epsilon>0$. We also prove that a previous index that had not been analyzed 
\cite{CM13}, but which has no useful worst-case time bounds for listing,
uses $O(n\log N +s\log^2 N)$ bits. As a byproduct, we offer a new variant
of a structure that finds the distinct values in an array range 
\cite{Mut02,Sad07}.

The general technique we use, of
augmenting the range search data structure used by grammar-based indexes, can 
be used for other kind of summarization queries. We illustrate this by 
providing the first grammar-based index
that uses $O(r \log N)$ bits, where $r$ is the size of a grammar that 
generates the text, and counts the number of occurrences of a pattern in time
$O(m^2+m\log^{2+\epsilon} r))$, for any constant $\epsilon>0$ (and
$O(m(\log N + \log^{2+\epsilon} r))$ if the grammar is balanced). We also
obtain the first Lempel-Ziv based index able of counting: if the text is
parsed into $z$ Lempel-Ziv phrases, then our index uses $O(z\log(N/z)\log N)$
bits and counts in time $O(m\log^{2+\epsilon} N)$. As a byproduct, we improve 
a previous result \cite{NNR13} on summing values over two-dimensional point 
ranges.

\paragraph*{Future work}

The space of our document listing index is an $O(\log D)$ factor away from 
what can be
expected from a grammar-based index. An important question is whether this
space factor can be removed or reduced while retaining worst-case time
guarantees for document listing. The analogous challenge in time is whether
we can get a time closer to the $\tilde O(m+\ndoc)$ that is obtained with
statistically-compressed indexes, instead of our $\tilde O(m \cdot \ndoc)$.

Another interesting question is whether there exists an index whose space and 
time can be bounded in terms of more general 
repetitiveness measures of the collection, for example in terms of the size $r$
of a grammar that represents the text, as is the case of grammar-based
pattern matching indexes that list all the occurrences of a pattern
\cite{CNfi10,CN12,CE18}. In particular, it would be interesting to handle
block edits, where a whole block of text is inserted, deleted, copied, or
moved. Such operations add only $O(\log N)$ nonterminals to a grammar, or
$O(1)$ phrases to a Lempel-Ziv parse, whereas our index can grow arbitrarily.

Yet another question is whether we can apply the idea of augmenting 
two-dimensional data structures in order to handle other kinds of summarization 
queries that are of interest in pattern matching and document retrieval 
\cite{Nav14}, for example counting the number of distinct documents where the
pattern appears, or retrieving the $k$ most important of those documents, or
retrieving the occurrences that are in a range of documents.


\appendix

\section{Proof of Correctness}
\label{sec:app}

We prove that our new document listing algorithm is correct. We first consider
a ``leftist'' algorithm that
proceeds as follows to find the distinct elements in $L[sp,ep]$.
It starts recursively with $[i,j] = [sp,ep]$ and remembers the documents that
have already been reported, globally. To process interval $[i,j]$, it reports
$L[i], L[i+1], \ldots$ until finding an already reported element at $L[d]$. 
Then it finds the minimum $E[k]$ in $E[d,j]$. If $L[k]$ had been reported
already, it stops; otherwise it reports $L[k]$ and proceeds recursively in 
$L[d,k-1]$ and $L[k+1,j]$, in this order. Our actual algorithm is a slight
variant of this procedure, and its correctness is established at the end.

\bigskip
\begin{lemma} \label{lem:alg}
The leftist algorithm reports the $\ndoc$ distinct elements in $L[sp,ep]$ in
$O(\ndoc)$ steps.
\end{lemma}
\begin{proof}
We prove that the
algorithm reports the leftmost occurrence in $L[sp,ep]$ of each distinct 
element. In particular, we prove by induction on $j-i$ that, when run on
any subrange $[i,j]$ of $[sp,ep]$, if (1) every leftmost occurrence in 
$L[sp,i-1]$ is already reported before processing $[i,j]$, then (2)
every leftmost occurrence in $L[sp,j]$ is reported after processing $[i,j]$.
Condition (1) holds for $[i,j]=[sp,ep]$, and we need to establish that (2)
holds after we process $[i,j]=[sp,ep]$.
The base case $i=j$ is trivial: 
the algorithm checks $L[i]$ and reports it if it was not reported before. 

On a larger interval $[i,j]$, the algorithm first reports $d-i$ occurrences of 
distinct elements in $L[i,d-1]$. Since these were not reported before, by
condition (1) they must be leftmost occurrences in $[sp,ep]$, and thus, after
reporting all the leftmost occurrences of $L[i,d-1]$, condition (1) holds for 
any range starting at $d$.

Now, we compute the position $k$ with minimum $E[k]$ in $E[d,j]$. Note that
$L[k]$ is a leftmost occurrence iff $E[k] < sp$, in which case it has not
been reported before and thus it should be reported by the algorithm. The 
algorithm, indeed, detects that it has not been reported before and therefore
recurses on $L[d,k-1]$, reports $L[k]$, and finally recurses 
on $L[k+1,j]$.\footnote{Since $L[k]$ does not appear in $L[d,k-1]$, the 
algorithm also works if $L[k]$ is reported before the recursive calls, which 
makes it real-time.}
Since those subintervals are inside $[i,j]$, we can apply induction.
In the call on $L[d,k-1]$, the invariant (1) holds and thus by induction we 
have that after the call the invariant (2) holds, so all the leftmost 
occurrences in $L[sp,k-1] = L[sp,d-1] \cdot L[d,k-1]$ have been reported. 
After we report $L[k]$ too, the invariant (1) also holds for the call
on $L[k+1,j]$, so by induction all the leftmost occurrences in $L[sp,j]$ have
been reported when the call returns.

In case $E[k] \ge sp$, $L[k]$ is not a leftmost occurrence in $L[sp,ep]$,
and moreover there are no leftmost occurrences in $L[d,j]$, so we should stop
since all the leftmost occurrences in $L[sp,j] = L[sp,d-1] \cdot L[d,j]$ are
already reported. Indeed, it must hold $sp \le E[k] < d$, since otherwise 
$E[E[k]] < E[k]$ and $d \le E[k] \le j$, contradicting the definition of
$k$. Therefore, by invariant (1), our algorithm already reported $L[k] =
L[E[k]]$, and hence it stops.

Then the algorithm is correct. As for the time, clearly the algorithm never 
reports the same element twice. The sequential part reports
$d-i$ documents in time $O(d-i+1)$. The extra $O(1)$ can be charged to the
caller, as well as the $O(1)$ cost of the subranges that do not produce any
result. Each calling procedure reports at least one element $L[k]$, so it 
can absorb those $O(1)$ costs, for a total cost of $O(\ndoc)$.
\end{proof}

Our actual algorithm is a variant of the leftist algorithm. When 
it takes the minimum $E[k]$ in $E[d,j]$, if $k=d$, it ignores that value and
takes instead $k=k'$, where $k'$ is some other value in $[d+1,j]$. Note that,
when processing $E[d,j]$ in the leftist algorithm, $L[d]$ is known to occur in 
$L[sp,d-1]$. Therefore, $E[d] \ge sp$, and if $k=d$, the leftist algorithm will
stop. The actual algorithm chooses instead position $k'$, but $E[k'] \ge E[d] 
\ge sp$, and therefore, as seen in the proof of Lemma~\ref{lem:alg}, the 
algorithm has already reported $L[k']$, and thus the actual
algorithm will also stop. Then the actual algorithm behaves identically to the
leftist algorithm, and thus it is also correct.

\bibliographystyle{plain}
\bibliography{paper}

\end{document}